\documentclass[12pt,twoside,a4paper,reqno]{amsart}
\usepackage[a4paper,width=160mm,top=25mm,bottom=30mm]{geometry}
\usepackage{amsmath,amssymb}
\usepackage{amsfonts}
\usepackage{enumerate}
\usepackage{graphicx}

\newtheorem{theorem}{Theorem}

\newtheorem{corollary}[theorem]{Corollary}

\newtheorem{definition}[theorem]{Definition}

\newtheorem{lemma}[theorem]{Lemma}

\newtheorem{proposition}[theorem]{Proposition}
\newtheorem{remark}[theorem]{Remark}

\begin{document}
\title{Lee-Yang Property and Gaussian Multiplicative chaos}
\author{Charles M. Newman and Wei Wu}
\address[Charles Newman]{Courant Institute of Mathematical Sciences, New
York University, 251 Mercer st, New York, NY 10012, USA \\
\& NYU-ECNU Institute of Mathematical Sciences at NYU Shanghai, 3663
Zhongshan Road North, Shanghai 200062, China}
\address[Wei Wu]{Courant Institute of Mathematical Sciences, New York
University, 251 Mercer st, New York, NY 10012, USA \\
\& Department of Statistics, University of Warwick, Coventry CV4 7AL, UK}
\maketitle

\begin{abstract}
The Lee-Yang property of certain moment generating functions having only
pure imaginary zeros is valid for Ising type models with one-component spins
and XY models with two-component spins. Villain models and complex Gaussian
multiplicative chaos are two-component systems analogous to XY\ models and
related to Gaussian free fields. Although the Lee-Yang property is known to
be valid generally in the first case, we show that is not so in the second.
Our proof is based on two theorems of general interest relating the Lee-Yang
property to distribution tail behavior.
\end{abstract}

\section{\protect\bigskip Overview}

The original theorem of Lee and Yang \cite{LY} on zeros of the partition
function of a spin-$\frac{1}{2}$ ($\left( \pm 1\right) -$valued) Ising model 
$\left\{ \sigma _{v}\right\} $ with pair ferromagnetic interactions implies
that the moment generating function $\mathbb{E}\left[ e^{zX}\right] $ of $%
X=\sum \lambda _{v}\sigma _{v}$ with $\lambda _{v}\geq 0$ has only pure
imaginary zeros (PIZ) in the complex $z-$plane. As we explain in more detail
below, this has been extended \cite{DN, LS} to $X=\sum \lambda _{v}S_{v}^{1}$
where $\left\{ S_{v}=\left( S_{v}^{1},S_{v}^{2}\right) \right\} $ for
certain two-component XY models where $S_{v}$ takes values in the unit
circle of $\mathbb{R}^{2}$. It also extends to $D-$component models with $%
D=3 $ \cite{DN}, known as classical Heisenberg models, but, as far as we
know, has not been extended to $D\geq 4$ --- see Remark 4 in \cite{DN}.

In this paper, we consider two interesting two-component systems, namely the
Villain model \cite{Vi} and complex Gaussian multiplicative chaos, related
to the Gaussian Free Field (GFF) \cite{LRV} and whether they have a Lee-Yang
PIZ property. For Villain models, like for XY models, such a property is
expressed in terms of the cosines of the angular variables; for Gaussian
multiplicative chaos, it is expressed in terms of the real part of the
complex-valued field and hence in terms of a sine-Gordon field --- since the
partition function of subcritical imaginary Gaussian multiplicative chaos in
an external field equals that of a sine-Gordon field. A PIZ property seems
plausible for both systems because of the connection between GFF and the XY
model via the spin-wave conjecture \cite{Dys, MW} on the one hand and the
connection between GFF and the Villain Gaussian-like distribution on the
other. Although Villain models are known to have the PIZ property (see p.
636 of \cite{FS0} and Theorem \ref{V} below) we prove that complex Gaussian
multiplicative chaos in general does not (see Propositions \ref{wpGMC} and %
\ref{noLY} below).

Our analysis relies on two results relating for (symmetric) random variables 
$X_{n}$, the PIZ\ property, tail behavior and convergence in distribution of 
$X_{n}$ to $X$, which may be of independent interest. The first of these
(Theorem \ref{weak}) is a new result which states that the PIZ\ property
plus a sub-Gaussian tail bound for each symmetrically distributed $X_{n}$
(with \emph{no uniformity} in $n$) implies not only that $X$ has the PIZ
property but also, surprisingly, sub-Gaussian tail behavior. The second
(Theorem \ref{slowtail2}) states that any $X$ with tail behavior strictly
between Gaussian and $\exp \left( -c\left\vert x\right\vert ^{1+\varepsilon
}\right) $ cannot have the PIZ\ property; it follows directly from a result
of Goldberg and Ostrovskii \cite{GO} (see Theorem 14.4.2 of \cite{L}).

\section{\protect\bigskip Introduction}

The Lee-Yang theorem, first obtained by Lee and Yang when studying phase
transitions in the classical Ising model, states that \textit{all} the
zeroes of the partition function of the Ising model as a function of the
external magnetic field lie on the imaginary axis \cite{LY}. Lee and Yang
proved this result only for the spin-$\frac{1}{2}$ (i.e., $\left( \pm
1\right) -$valued) Ising model, but later it was extended by Griffiths to
spin-$\frac{n}{2}$ models \cite{Gr}, and then by Newman to ferromagnetic
Ising models with quite general single spin distributions \cite{Ne} (see
also \cite{LS}). Since then, it has been applied to prove the properties of
the infinite volume limit and existence of a mass gap under an external
magnetic field \cite{LP, GRS, FR2}, and to prove correlation inequalities 
\cite{Ne2, Dun2} --- see also \cite{FR} for a general review of the Lee-Yang
type theorem and their applications.

For two-component ferromagnets, analogous Lee-Yang type theorems were first
proved by Dunlop and Newman \cite{DN} for the classical XY\ model (see also 
\cite{SF}) and then by Lieb and Sokal \cite{LS} for generic two-component
ferromagnets with quite general single spin distributions. We now recall the
Lee-Yang theorem for the classical XY model. Given a finite graph $G=\left( 
\mathcal{V}\left( G\right) ,\mathcal{E}\left( G\right) \right) $, we denote
the spin variable at each $v\in \mathcal{V}\left( G\right) $ by $%
S_{v}=\left( S_{v}^{1},S_{v}^{2}\right) $ in the unit circle. Given real
numbers $\left\{ J_{e}\right\} _{e\in \mathcal{E}}$, the classical XY model
on $G$ is defined by a Gibbs measure 
\begin{equation}
Z_{G}^{-1}\exp \left( -\frac{1}{T}H\left( S\right) \right) \prod_{v\in 
\mathcal{V}\left( G\right) }\delta (\left\vert S_{v}\right\vert =1),
\label{cXY}
\end{equation}%
with the Hamiltonian $H\left( S\right) $ given by 
\begin{equation}
H\left( S\right) =-\sum_{e=\left( i,j\right) \in \mathcal{E}\left( G\right)
}J_{e}\text{ }S_{i}\cdot S_{j},  \label{XYham}
\end{equation}%
where $S_{i}\cdot S_{j}=S_{i}^{1}S_{j}^{1}+S_{i}^{2}S_{j}^{2}$, and with $\
Z_{G}=Z_{G}\left( T,\left\{ J_{e}\right\} \right) $ the normalization
constant that makes (\ref{cXY}) a probability measure. This is an $XY$ model
on $G$ with free boundary conditions; for a discussion of other boundary
conditions where the conclusions of the next theorem remain valid, see
Remark \ref{bc} below.

\begin{theorem}[\protect\cite{DN, LS}]
\label{LYXY}Suppose that $J_{e}\geq 0$ for all $e\in \mathcal{E}$ and $%
\lambda _{v}\geq 0$ for all $v\in \mathcal{V}(G)$, then 
\begin{equation*}
f_{XY}\left( z\right) :=\mathbb{E}\left[ \exp \left( z\sum_{v\in \mathcal{V}(G)%
}\lambda _{v}S_{v}^{1}\right) \right]
\end{equation*}%
has only pure imaginary zeros (namely, it is not zero when $\Re{z}>0$ or $\Re%
{z}<0$). Here $\mathbb{E}$ dentoes the expectation with respect to the
probability measure (\ref{cXY}).
\end{theorem}

\begin{remark}
Theorem \ref{LYXY} is a corollary of a more general result that $\mathbb{E}%
\left[ \exp \left( \sum_{v\in \mathcal{V}(G)}z_{v}S_{v}^{1}\right) \right] $
does not vanish when $\Re{z_{v}}>0$ $\forall v\in \mathcal{V}(G).$
\end{remark}

The classical Villain model is another two-component spin model which is
closely related to the XY\ model \cite{Vi}. Given any finite graph $G=\left( 
\mathcal{V}\left( G\right) ,\mathcal{E}\left( G\right) \right) $, and
positive numbers $\left\{ J_{e}\right\} _{e\in \mathcal{E}\left( G\right) }$%
, a Villain model on $G$ is defined by the Gibbs measure 
\begin{equation}
Z_{G}^{-1}\prod_{e=\left( u,v\right) \in \mathcal{E}\left( G\right)
}V_{e}\left( \theta _{u}-\theta _{v}\right) \prod_{v\in \mathcal{V}\left(
G\right) }d\theta _{v},  \label{Vil}
\end{equation}%
where $\theta _{v}\in (-\pi ,\pi ]$ for $v\in \mathcal{V}\left( G\right) $
and $d\theta _{v}$ is Lebesgue measure on $(-\pi ,\pi ]$, 
\begin{equation}
V_{e}\left( \theta \right) =\sum_{m\in \mathbb{Z}}\exp \left( -\frac{J_{e}}{2%
}\left( \theta +2\pi m\right) ^{2}\right)   \label{Ve}
\end{equation}%
is a periodized Gaussian, and $Z_{G}$ is the normalizing constant. Let $%
\left( \Theta _{v}:v\in \mathcal{V}\left( G\right) \right) $ be jointly
distributed by the Gibbs measure (\ref{Vil}). Given non-negative real
numbers $\left\{ \lambda _{v}\right\} _{v\in \mathcal{V}\left( G\right) }$,
let $\mu _{\lambda }$ denote the distribution of 
\begin{equation}
\sum_{v\in \mathcal{V}\left( G\right) }\lambda _{v}\cos \Theta _{v}.
\label{lam}
\end{equation}

The validity of the Lee-Yang type property as well as correlation
inequalities for Villain models, due to Fr{\"o}hlich and Spencer and to
Bellissard (\cite{FS0} and ref. 34 there), is discussed and a proof is
sketched on pp. 635-636 of \cite{FS0}. Since that Lee-Yang property does not
seem to be widely known, we present it as the next theorem and provide a
detailed proof in Sections \ref{Sec:V} and \ref{Sec:pf}. The Lee-Yang
property does not seem to follow from the Lieb-Sokal approach since unlike (%
\ref{XYham}), the Villain model (\ref{Vil}) is not written as a Gibbs
measure with ferromagnetic pair interactions.

\begin{theorem}[\protect\cite{FS0}]
\label{V}Suppose that $J_{e}\geq 0$ for all $e\in \mathcal{E}\left( G\right) 
$ and $\lambda _{v}\geq 0$ for all $v\in \mathcal{V}\left( G\right) $, then 
\begin{equation*}
f_{V}\left( z\right) :=\mathbb{E}\left[ \exp \left( z\sum_{v\in \mathcal{V}%
\left( G\right) }\lambda _{v}\cos \Theta _{v}\right) \right] =\int_{-\infty
}^{\infty }e^{zx}d\mu _{\lambda }\left( x\right)
\end{equation*}%
has only pure imaginary zeros as a function of complex $z$. Here $\mathbb{E}$
dentoes the expectation with respect to the probability measure (\ref{Vil}).
\end{theorem}

\begin{remark}
By extending Theorem \ref{weak} below to a multivariable version, the same
proof leads to a more general result, that is, for any $\left\{
z_{v}\right\} _{v\in \mathcal{V}\left( G\right) }$ such that $\Re{z_{v}}>0$ $%
\forall v$, \newline
$\mathbb{E}\left[ \exp \sum_{v\in \mathcal{V}\left( G\right) }z_{v}\cos
\Theta _{v}\right] \neq 0$.
\end{remark}

\begin{remark}
\label{bc}Theorem \ref{V} is stated for the Villain model (\ref{Vil}) with
free boundary condition. The same result holds for some other boundary
conditions, including periodic (i.e., the Villain model on a torus), or for
Dirichlet-type boundary condition where $\Theta _{v}$ for all $v$ in the
boundary are equal with a value that is uniformly distributed on the unit
circle. Periodic boundary condition on a torus simply correspond to a
different choice of $G$ than do free boundary conditions, while Dirichlet
type boundary conditions can be handled as a corollary of Theorem \ref{V} by
adding couplings between boundary vertices and letting the coupling
magnitudes tend to infinity. We note that the Lee-Yang property of Theorem %
\ref{LYXY} for XY models will also be valid for such boundary conditions.
\end{remark}

Obtaining new Lee-Yang type theorems for lattice models can be useful to
derive new results for continuum field theories. For example, Simon and
Griffiths \cite{SG} proved a Lee-Yang result for the continuum $\left( \phi
^{4}\right) _{2}$ Euclidean field theory by first obtaining a new Lee-Yang
result for $\phi ^{4}$ lattice models. Indeed, part of our motivation comes
from the so-called spin-wave conjecture \cite{Dys, MW}, which states that
for both the XY and the Villain model at temperature $T$ less than some
critical value, on large scales the angular variables $\Theta _{v}$ behave
like a Gaussian Free Field (GFF) modulo $2\pi $. This suggests that the spin
field ($S_{\cdot }$ in the XY\ model and $\exp \left( i\Theta _{\cdot
}\right) $ in the Villain model) may behave like a version of complex
Gaussian multiplicative chaos (see e.g., \cite{LRV}); or, after a duality
transformation, a version of the Sine-Gordon field \cite{FS}. We have not
obtained a Lee-Yang property for complex Gaussian multiplicative chaos and
even ruled it out (as we discuss below) in a certain parameter range.
However that parameter range does not correspond to very low temperature $T$%
, so it may still be that there is a low $T$ Lee-Yang property.

Our approach is based on showing that a certain set of properties of moment
generating functions (see Definition \ref{LYp}) is preserved under
convergence in distribution (Theorem \ref{weak}). This is more than the
requirement that all zeros are pure imaginary --- it further requires a
sub-Gaussian tail for the distribution. In Sections \ref{Sec:V} and \ref%
{Sec:pf} we prove Theorem \ref{V} by approximating the Villain model\ on any
finite graph by one dimensional XY spin chains. Since the XY models do
satisfy the Lee-Yang property, we obtain the result by applying the
convergence Theorem \ref{weak}.

Based on the spin-wave conjecture, it is natural to ask whether there is a
Lee-Yang property for complex Gaussian multiplicative chaos, namely $\exp
\left( i\beta h\right) $ where $h$ is a two dimensional GFF. In Section \ref%
{Sec:GMC} we give a negative answer to this question when $\beta \in \left(
1,\sqrt{2}\right) $ --- the Lee-Yang property does not hold for complex
Gaussian multiplicative chaos with such values of $\beta $. The proof is
based on a general theorem that shows that the pure imaginary zeros property
does not hold for random variables with tail slower than Gaussian (Theorem %
\ref{slowtail2}), and an explicit computation of the tail probability for
the\ integral of complex Gaussian multiplicative chaos. We also study the
so-called discrete complex Gaussian multiplicative chaos on finite graphs,
which roughly speaking, is $\exp \left( i\beta h\right) $ where $h$ is a
discrete Gaussian free field (DGFF) with certain boundary condition where
the Lee-Yang property might be expected to hold. We show (see Proposition %
\ref{noLY}) that for discrete complex Gaussian multiplicative chaos with $%
\beta \in \left( 1,\sqrt{2}\right) $, the Lee-Yang property cannot hold on
all finite graphs. Interestingly, we use there a corollary of the weak
convergence result of Theorem \ref{weak} to rule out the Lee-Yang property.

Complex Gaussian multiplicative chaos corresponds to a type of a Sine-Gordon
field (see, e.g., \cite{FS} and \cite{DH}). It has been pointed out to us by
T. Spencer [personal communication] that because of the relation of the
Lee-Yang property to exponential decay or existence of a mass gap for
Sine-Gordon fields, how the validity of the Lee-Yang property depends on the
parameter $\beta $ is of some interest. The Sine-Gordon beta-parameter used
in \cite{FS}, which we denote here by $\tilde{\beta}$, is related to the $%
\beta $ we use for complex Gaussian multiplicative chaos by $\tilde{\beta}%
=2\pi \beta ^{2}$. Thus our region $\beta \in \left( 1,\sqrt{2}\right) $ of
non-validity of the Lee-Yang property corresponds to $\tilde{\beta}\in
\left( 2\pi ,4\pi \right) $.

\section{Lee-Yang property and weak convergence\label{sec2}}

Let $\mu $ be a probability measure on $\mathbb{R}$ and $X$ be a random
variable on some probability space $\left( \Omega ,\mathcal{F},\mathbb{P}%
\right) $ with distribution $\mu $.

\begin{definition}
\label{LYp}We say $\mu $ (or $X$) is of Lee-Yang type (and write $\mu \in 
\mathcal{L}$) if
\end{definition}

\begin{enumerate}
\item $X$ has the same distribution as $-X$

\item $\mathbb{E}\left[ \exp \left( bX^{2}\right) \right] <\infty $ for some 
$b>0$

\item For $z\in \mathbb{C}$, $\mathbb{E}\left[ \exp \left( zX\right) \right] 
$ only has zeros on the pure imaginary axis.
\end{enumerate}

The next theorem states that the Lee-Yang type property is preserved under
weak convergence and helps explain why the sub-Gaussian property (2) is
built into Definition ~\ref{LYp}.

\begin{theorem}
\label{weak}Suppose for each $n\in \mathbb{N}$, $\mu _{n}\in \mathcal{L}$,
and $\mu _{n}$ converges weakly to the probability measure $\mu $. Then $\mu
\in \mathcal{L}$.
\end{theorem}

\begin{remark}
We note in particular that the limiting measure must satisfy $\mu \left(
r,\infty \right) \leq \exp \left( -cr^{2}\right) $ for some $c>0$. We will
use this fact later. This property of $\mu $ seems a-priori surprising,
because if one does not assume $\mu _{n}$ has the pure imaginary zero
property (Definition \ref{LYp},(3)), the conclusion is not true since the
constant $b$ in Definition \ref{LYp},(2) may depend on\ $n$.
\end{remark}

\begin{remark}
Approximation schemes play a major role in classical work (e.g., \cite{SG,
Ne}) on Lee-Yang type theorems for Ising-like systems. It seems that Theorem %
\ref{weak} fits together with that work primarily as a helpful tool ---
i.e., it suffices to show weak convergence, without any extra moment or
moment generating function estimates, to guarantee that the limit system
will have the desired Lee-Yang property.
\end{remark}

\begin{corollary}
\label{seq} Suppose that for $n=1,2,...,$

\begin{enumerate}
\item $X_{n}$ has the same distribution as $-X_{n}$, and

\item $\mathbb{E}\left[ \exp \left( b_{n}X_{n}^{2}\right) \right] <\infty $
for some $b_{n}>0$.
\end{enumerate}

\noindent If $X_{n}$ converges in distribution to $X$ and $\mathbb{E}\left[
\exp \left( bX^{2}\right) \right] =\infty $ for all $b>0$, it follows that
for all but finitely many $n$, $\mathbb{E}\left[ \exp \left( zX_{n}\right) %
\right] $ has some zeros that are not purely imaginary.
\end{corollary}

\begin{proof}
The proof is by contradiction. If the conclusion were not valid, then there
would be a subsequence $X_{n_{k}}$ with distributions $\mu _{k}\in \mathcal{L%
}$ which would converge weakly to the distribution $\mu $ of $X$ with $\mu
\notin \mathcal{L}$. That would contradict Theorem \ref{weak}, which
completes the proof.
\end{proof}

The next theorem relates the Lee-Yang property to the distribution tail
behavior and explains further why the sub-Gaussian property (2) is natural
in Definition \ref{LYp}. As we explain below, it follows directly from a
theorem of Goldberg and Ostrovskii \cite{GO} (see Theorem 14.4.2 of \cite{L}%
).

\begin{theorem}
\label{slowtail2} Suppose the random variable $X$ satistifes the following
two properties:

\begin{enumerate}
\item $\mathbb{E}e^{b\left\vert X\right\vert ^{a}}<\infty $ for some $b>0$
and $a>1,$

\item $\mathbb{E}e^{b^{\prime }X^{2}}=\infty $ for all $b^{\prime }>0$.
\end{enumerate}

\noindent Then $\mathbb{E}e^{zX}$ has some zeroes that are not purely
imaginary.
\end{theorem}

\begin{remark}
A natural question is how much Property (1) of this theorem can be weakened
without changing the conclusions. The answer is not very much. This can be
seen by constructing examples of random variables $X$ where $\mathbb{E}%
\left( e^{b\left\vert X\right\vert \log \left\vert X\right\vert }\right)
<\infty $ for some $b>0$ but $\mathbb{E}\left[ e^{zX}\right] $ has no zeros
that are not purely imaginary. Probably the simplest example is a Poisson
random variable $X$, where $\mathbb{E}\left[ e^{zX}\right] =\exp \left(
\lambda \left( e^{z}-1\right) \right) $ has no zeros at all. One can also
extend this example to obtain symmetric random variables with Poisson-type
tail behavior whose moment generating functions do have many zeros, all
purely imaginary.
\end{remark}

We now turn to the proof of Theorems \ref{weak} and \ref{slowtail2}. The
proof of Theorem \ref{weak} is based on the uniform convergence of entire
functions and an application of Hurwitz Theorem. The proof will be given in
two steps. In the first step we prove Theorem \ref{weak} under the
additional assumption\ that $\sup_{n}\mathbb{E}\left[ X_{n}^{2}\right]
<\infty $. In the second step we show $\sup_{n}\mathbb{E}\left[ X_{n}^{2}%
\right] <\infty $ automatically holds.

The key ingredient for the proof is the following Proposition (see
Proposition 2 of \cite{Ne2}). We include a proof for completeness.

\begin{proposition}[\protect\cite{Ne2}]
\label{entire}Suppose $X\in \mathcal{L}$. Then $f\left( z\right) :=\mathbb{E}%
\left[ \exp \left( zX\right) \right] $ is an entire function of $z$ with
product expansion 
\begin{equation*}
f\left( z\right) =e^{Bz^{2}}\prod_{k}\left( 1+\frac{z^{2}}{y_{k}^{2}}\right)
,
\end{equation*}%
where $B\geq 0$, $y_{k}\in \mathbb{R}$ such that $\sum_{k}\frac{1}{y_{k}^{2}}%
<\infty $. Also $\mathbb{E}\left[ X^{2}\right] =$Var$\left[ X\right]
=2\left( B+\sum_{k}\frac{1}{y_{k}^{2}}\right) $.
\end{proposition}

\begin{proof}
We first note that $f$ is entire of (exponential) order $2$ (and finite
type), since $zX=bX^{2}-b\left( X-z/2b\right) ^{2}+z^{2}/4b$ and so 
\begin{equation*}
\left\vert f\left( z\right) \right\vert \leq \mathbb{E}\left[ \exp \left(
bX^{2}+\frac{\left\vert z\right\vert ^{2}}{4b}\right) \right] .
\end{equation*}%
By (2) of Definition \ref{LYp}, there is some $C<\infty $ such that $%
\left\vert f\left( z\right) \right\vert \leq C\exp \left( \left( 4b\right)
^{-1}\left\vert z\right\vert ^{2}\right) $. By the Hadamard factorization
theorem, 
\begin{equation*}
f\left( z\right) =e^{P_{2}\left( z\right) }z^{m_{0}}\prod_{j}\left( 1-\frac{z%
}{z_{j}}\right) e^{z/z_{j}},
\end{equation*}%
where $P_{2}$ is a quadratic polynomial, $m_{0}$ is the degree of zero of $f$
at the origin, $\left\{ z_{j}\right\} $ are the other zeros and $%
\sum_{j}\left\vert z_{j}\right\vert ^{-2}<\infty $. Since $\mu $ is a
symmetric probability measure, $f\left( z\right) =f\left( -z\right) $ with $%
f\left( 0\right) =1$, and $f$ only has pure imaginary zeros, we have $%
m_{0}=0 $, $P_{2}\left( z\right) =Bz^{2}$ and $\left\{ z_{j}\right\} $ come
in pairs $\left\{ \pm iy_{k}\right\} $. Combining the pairs gives 
\begin{equation*}
f\left( z\right) =e^{Bz^{2}}\prod_{k}\left( 1+\frac{z^{2}}{y_{k}^{2}}\right)
\end{equation*}%
and 
\begin{equation*}
\mathbb{E}\left[ X^{2}\right] =f^{\prime \prime }\left( 0\right) =2\left(
B+\sum_{k}\frac{1}{y_{k}^{2}}\right) .
\end{equation*}
\end{proof}

\begin{proof}[Proof of Theorem \protect\ref{weak}]
Let $f_{n}\left( z\right) =\mathbb{E}\left[ \exp \left( zX_{n}\right) \right]
$. We first prove Theorem \ref{weak} assuming\newline
$\sup_{n}\mathbb{E}\left[ X_{n}^{2}\right] <\infty $. We claim that it
suffices to prove 
\begin{equation}
\sup_{n}\left\vert f_{n}\left( z\right) \right\vert <\infty \text{ uniformly
on compact sets of }z  \label{unif1}
\end{equation}%
and that%
\begin{equation}
\sup_{n}\mathbb{E}\left[ \exp \left( b^{\prime }X_{n}^{2}\right) \right]
<\infty \text{ for some fixed }b^{\prime }>0\text{.}  \label{unif2}
\end{equation}%
We now explain why (\ref{unif1}) and (\ref{unif2}) suffice to imply the
conclusion of Theorem \ref{weak}. First note that the validity of (1) of
Definition \ref{LYp} for each $\mu _{n}$ implies it for $\mu $. As $X_{n}$
converges in distribution to $X$, $f_{n}\rightarrow f$ on the pure imaginary
axis, and (\ref{unif1}) implies that $f$ extends to an entire function with $%
f_{n}\rightarrow f$ uniformly on compact sets. Moreover, by Hurwitz'
Theorem, open zero-free regions for all $f_{n}$ (e.g., $\mathbb{C}\backslash 
\mathbf{i}\mathbb{R}$) are zero-free for $f$. This verifies (3) of
Definition \ref{LYp} for $X$. Finally, (\ref{unif2}) implies (2) of
Definition \ref{LYp} for $X$ (e.g., by taking any $b\in \left( 0,b^{\prime
}\right) $) and thus $X\in \mathcal{L}$.

We next claim that (\ref{unif1}) and (\ref{unif2}) are direct consequences
of Proposition \ref{entire}. Apply Proposition \ref{entire} and use the fact
that $\left\vert 1+z^{2}/y^{2}\right\vert \leq \exp \left( \left\vert
z\right\vert ^{2}/y^{2}\right) $ to see that 
\begin{equation}
\left\vert f_{n}\left( z\right) \right\vert \leq \exp \left[ \left(
B^{\left( n\right) }+\sum_{k}\frac{1}{\left( y_{k}^{\left( n\right) }\right)
^{2}}\right) \left\vert z\right\vert ^{2}\right] =\exp \left( \frac{1}{2}%
\text{Var}\left[ X_{n}\right] \left\vert z\right\vert ^{2}\right) ,
\label{mgf}
\end{equation}%
where $B^{\left( n\right) }$ and $\left\{ y_{k}^{\left( n\right) }\right\} $
for $X_{n}$ correspond to $B$ and $\left\{ y_{k}\right\} $ for $X$ in
Proposition \ref{entire}. Since we assumed $\sup_{n}\mathbb{E}\left[
X_{n}^{2}\right] <\infty $, we conclude (\ref{unif1}). To prove (\ref{unif2}%
), note that (\ref{mgf}) implies that the tail of $X_{n}$ is dominated by
the tail of $Y_{n}\sim \mathcal{N}\left( 0,\text{Var}\left[ X_{n}\right]
\right) $. Therefore we conclude (\ref{unif2}) with any $b^{\prime }<\left(
2\sup_{n}\mathbb{E}\left[ X_{n}^{2}\right] \right) ^{-1}$.

Finally we prove that convergence of $X_{n}$ to some $X$ in distribution
implies that \newline
$\sup_{n}\mathbb{E}\left[ X_{n}^{2}\right] <\infty $. We will argue by
contradiction. Suppose that $X_{n}\in \mathcal{L}$, $\sup_{n}\mathbb{E}\left[
X_{n}^{2}\right] =\infty $ and $X_{n}$ converges to some $X$ in
distribution. By taking a subsequence and applying Proposition \ref{entire}
we may assume that 
\begin{equation}
\frac{1}{2}\mathbb{E}\left[ X_{n}^{2}\right] =B^{\left( n\right) }+\sum_{k}%
\frac{1}{\left( y_{k}^{\left( n\right) }\right) ^{2}}\rightarrow \infty 
\text{.}  \label{div}
\end{equation}%
We also know from the convergence in distribution of $X_{n}$ that $%
f_{n}\left( it\right) \rightarrow f\left( it\right) $ uniformly on compact
subsets of $t\in \mathbb{R}$. Both $f_{n}\left( it\right) $ and $f\left(
it\right) $ are real and continuous in $t$ so that there exists $\varepsilon
>0$ such that $f\left( it\right) >0$ for $t\in \left[ 0,\varepsilon \right] $%
. Since $f_{n}\left( iy_{1}^{\left( n\right) }\right) =0$, we must have $%
\liminf_{n\rightarrow \infty }y_{1}^{\left( n\right) }\geq \varepsilon $.
However, by Proposition \ref{entire}, for $t\in (0,y_{1}^{\left( n\right) }]$%
, 
\begin{equation*}
f_{n}\left( it\right) =e^{-B^{\left( n\right) }t^{2}}\prod_{k}\left( 1-\frac{%
t^{2}}{\left( y_{k}^{\left( n\right) }\right) ^{2}}\right) \leq \exp \left(
\left( -B^{\left( n\right) }-\sum_{k}\frac{1}{\left( y_{k}^{\left( n\right)
}\right) ^{2}}\right) t^{2}\right) .
\end{equation*}%
By (\ref{div}) this goes to zero as $n\rightarrow \infty $. This contradicts 
$f\left( it\right) >0$ for $t\in \left[ 0,\varepsilon \right] $ and
completes the proof.
\end{proof}

Finally we note that Theorem \ref{slowtail2} is an immediate consequence of
the following proposition, which follows from the Goldberg-Ostrovskii result 
\cite{GO} stated as Theorem 14.4.2 in \cite{L}. We\ also note that there is
a typographical error in \cite{L} and $\sum_{k}a_{k}^{2}<\infty $ should be
replaced there by $\sum_{k}\left( 1/a_{k}^{2}\right) <\infty $.

\begin{proposition}
\label{slowtail} Suppose that the random variable $Y$ satisfies Property 1
of Theorem \ref{slowtail2}, and $\mathbb{E}e^{zY}$ has only pure imaginary
zeroes. Then $Y-\mathbb{E}Y$ belongs to the class $\mathcal{L}$.
\end{proposition}

\begin{proof}
The proof follows directly from a result of Goldberg and Ostrovskii (see
Theorem 14.4.2 of \cite{L}). Arguing as in the proof of Proposition \ref%
{entire}, let $f\left( z\right) =\mathbb{E}e^{zY}$; then by Young's
inequality, 
\begin{equation*}
\left\vert f\left( z\right) \right\vert \leq \mathbb{E}\left[ \exp \left(
b\left\vert Y\right\vert ^{a}+\frac{1}{aa^{\prime }b}\left\vert z\right\vert
^{a^{\prime }}\right) \right] \text{, \ where }a^{\prime }=\left(
1-a^{-1}\right) ^{-1}.
\end{equation*}%
By Property 1 of Theorem \ref{slowtail2} we have that 
\begin{equation*}
\left\vert f\left( z\right) \right\vert \leq C\exp \left( \frac{1}{%
aa^{\prime }b}\left\vert z\right\vert ^{a^{\prime }}\right) \text{, for some 
}C<\infty \text{.}
\end{equation*}%
Therefore $f$ is an entire function of finite (exponential) order $a^{\prime
}$.

By the Goldberg-Ostrovskii result it follows that%
\begin{equation*}
f\left( z\right) =e^{\alpha z}e^{Bz^{2}}\prod_{k}\left( 1+\left( \frac{z}{%
y_{k}}\right) ^{2}\right) ,
\end{equation*}%
with $\alpha \in \mathbb{R}$, $B\geq 0$ and $\sum_{k}\left\vert
y_{k}\right\vert ^{-2}<\infty $. But this implies that $\mathbb{E}Y=\alpha $
and $Y-\mathbb{E}Y$ satisfies all the properties to be in $\mathcal{L}$.
\end{proof}

\section{$1D$ XY Spin Chain\label{Sec:V}}

In this and the next sections we will show that the distribution of (\ref%
{lam}) from the Villain model is of Lee-Yang type as defined in Definition %
\ref{LYp}, which implies Theorem \ref{V}. The proof uses the following
scaling limit result for a one dimensional XY model.

Let $G_{n}$ denote the graph whose vertex set $\mathcal{V}\left(
G_{n}\right) $ is $\left\{ 0,1/n,2/n,...,1\right\} $ and edge set is $%
\left\{ \left\{ \left( j-1\right) /n,j/n\right\} :j=1,...,n\right\} $. We
assign to each $i\in \mathcal{V}\left( G_{n}\right) $ a spin variable $%
S_{i}^{n}$ in the unit circle, with the corresponding angle $\theta
_{i}^{n}\in (-\pi ,\pi ]$, and consider the XY model on $G_{n}$ defined by
the Gibbs measure 
\begin{equation}
Z_{G_{n}}^{-1}\exp \left( -\frac{1}{T_{n}}H\left( S^{n}\right) \right)
\prod_{i\in \mathcal{V}(G_{n})}\delta (\left\vert S_{i}^{n}\right\vert =1),
\label{xy}
\end{equation}%
with Hamiltonian 
\begin{equation}
H\left( S^{n}\right) =-\sum_{\left( i,j\right) \in \mathcal{E}\left(
G_{n}\right) }S_{i}^{n}\cdot S_{j}^{n}=-\sum_{\left( i,j\right) \in \mathcal{%
E}\left( G_{n}\right) }\cos \left( \theta _{i}^{n}-\theta _{j}^{n}\right) .
\label{H1dXY}
\end{equation}%
For the remainder of this section, we write $\mathcal{B}$ or $\mathcal{B}%
_{n} $ for $1/T$ or $1/T_{n}$.

We note that in the next proposition, the parameter $b>0$ will eventually be
proportional to one of the $J_{e}$'s in (\ref{Ve}) of the Villain model ---
see Remark \ref{heat} below.

\begin{proposition}
\label{1dXY}If $\mathcal{B}_{n}/n\rightarrow b\in \left( 0,\infty \right) $,
then $S_{\left[ nt\right] }^{n}$ converges in distribution (using a Skorohod
metric) to $S\left( t\right) =\left( \cos \Phi \left( t\right) ,\sin \Phi
\left( t\right) \right) $, where $\Phi \left( 0\right) $ is uniformly
distributed in $[-\pi ,\pi )$ and $\exp \left( i\left( \Phi \left( t\right)
-\Phi \left( 0\right) \right) \right) $ is distributed as $\exp \left(
iB\left( t\right) /\sqrt{b}\right) $, where $B\left( t\right) $ is a
standard one-dimensional Brownian motion.
\end{proposition}

\begin{remark}
\label{heat}This implies that the probability density on $(-\pi ,\pi ]$ of $%
\Phi \left( t\right) -\Phi \left( 0\right) $ is 
\begin{equation*}
\sum_{m\in \mathbb{Z}}\frac{1}{\sqrt{2\pi t}}\exp \left( -\frac{b\left(
\theta +2\pi m\right) ^{2}}{2t}\right) ,
\end{equation*}%
which is proportional to $V_{e}$ defined in (\ref{Ve}) with $J_{e}=b/t$.
\end{remark}

\begin{proof}
One can view $\theta ^{n}$ with $S^{n}=\left( \cos \theta ^{n},\sin \theta
^{n}\right) $ as a one dimensional Markov process with transition density
given by 
\begin{equation*}
K_{\mathcal{B}}\left( \theta ,\theta ^{\prime }\right) =\frac{e^{\mathcal{B}%
\cos \left( \theta ^{\prime }-\theta \right) }}{\int_{-\pi }^{\pi }e^{%
\mathcal{B}\cos \phi }d\phi }\text{ \ for }\theta ,\theta ^{\prime }\in
(-\pi ,\pi ]
\end{equation*}%
and initial distribution uniform on $(-\pi ,\pi ]$. We provide a sketch of a
proof of the scaling limit result. Based on a standard convergence result
for discrete time Markov chains (see, e.g., Theorem 17.28 of \cite{Kal}), it
suffices to prove that for any twice differentiable function $f:\mathbb{S}%
^{1}\mapsto \mathbb{R}$, the transition operator $\mathcal{K}_{\mathcal{B}}$
defined by%
\begin{equation*}
\mathcal{K}_{\mathcal{B}}f\left( \theta \right) =\int_{-\pi }^{\pi }K_{%
\mathcal{B}}\left( \theta ,\theta ^{\prime }\right) f\left( \theta ^{\prime
}\right) d\theta ^{\prime },
\end{equation*}%
satisfies 
\begin{equation*}
\frac{1}{1/n}\left( \mathcal{K}_{\mathcal{B}_{n}}-I\right) f\rightarrow 
\frac{1}{2b}\frac{d^{2}}{d\theta ^{2}}f.
\end{equation*}%
Notice that $\frac{1}{2b}\frac{d^{2}}{d\theta ^{2}}$ is the generator of $%
\exp \left( iB\left( t\right) /\sqrt{b}\right) $.

Indeed, since $\mathcal{B}_{n}=O\left( n\right) $, Laplace's method (see
e.g., \cite{Erd}) yields 
\begin{equation*}
\int_{-\pi }^{\pi }e^{\mathcal{B}_{n}\cos \phi }d\phi =e^{\mathcal{B}_{n}}%
\sqrt{\frac{2\pi }{\mathcal{B}_{n}}}\left( 1+\frac{1}{8\mathcal{B}_{n}}%
+O\left( \frac{1}{\mathcal{B}_{n}^{2}}\right) \right) .
\end{equation*}%
Also, the value of $\theta ^{\prime }$ that minimizes $K_{\mathcal{B}}\left(
\theta ,\theta ^{\prime }\right) f\left( \theta ^{\prime }\right) $ is given
by 
\begin{equation*}
-\sin \theta _{m}^{^{\prime }}+\frac{1}{\mathcal{B}_{n}}\frac{f^{\prime
}\left( \theta _{m}^{\prime }-\theta \right) }{f\left( \theta _{m}-\theta
\right) }=0,
\end{equation*}%
or 
\begin{equation*}
\theta _{m}^{^{\prime }}=\theta +\frac{1}{\mathcal{B}_{n}}\frac{f^{\prime
}\left( \theta \right) }{f\left( \theta \right) }+O\left( \frac{1}{\mathcal{B%
}_{n}^{2}}\right) .
\end{equation*}%
Let $g\left( \theta ^{\prime }\right) =\mathcal{B}_{n}^{-1}\log \left[ \exp
\left( \mathcal{B}_{n}\cos \left( \theta ^{\prime }-\theta \right) \right)
f\left( \theta ^{\prime }\right) \right] $, then Laplace's method yields%
\begin{equation*}
\int_{-\pi }^{\pi }e^{\mathcal{B}_{n}g\left( \theta ^{\prime }\right)
}d\theta ^{\prime }=e^{\mathcal{B}_{n}\cos \left( \theta _{m}^{\prime
}-\theta \right) }f\left( \theta _{m}^{\prime }\right) \sqrt{\frac{2\pi }{%
\mathcal{B}_{n}g^{\prime \prime }\left( \theta _{m}^{\prime }\right) }}%
\left( 1+\frac{1}{8\mathcal{B}_{n}}+O\left( \frac{1}{\mathcal{B}_{n}^{2}}%
\right) \right) .
\end{equation*}%
In particular, for any $f\in H^{1}\left( \mathbb{S}^{1}\right) $, 
\begin{eqnarray*}
\mathcal{K}_{\mathcal{B}_{n}}f\left( \theta \right)  &=&f\left( \theta
_{m}^{\prime }\right) e^{-\frac{1}{2\mathcal{B}_{n}}\left( f^{\prime
}/f\right) ^{2}}\left( 1-\frac{1}{\mathcal{B}_{n}}\frac{f^{\prime \prime
}\left( \theta \right) }{f\left( \theta \right) }+\frac{1}{\mathcal{B}_{n}}%
\left( \frac{f^{\prime }\left( \theta \right) }{f\left( \theta \right) }%
\right) ^{2}\right) ^{-1/2}\left( 1+O\left( \frac{1}{\mathcal{B}_{n}^{2}}%
\right) \right)  \\
&=&f\left( \theta \right) +\frac{1}{2\mathcal{B}_{n}}f^{\prime \prime
}\left( \theta \right) +O\left( \frac{1}{\mathcal{B}_{n}^{2}}\right) ,
\end{eqnarray*}%
which implies 
\begin{equation*}
\frac{1}{1/n}\left( \mathcal{K}_{\mathcal{B}_{n}}-I\right) f\rightarrow 
\frac{1}{2b}f^{\prime \prime }\left( \theta \right) .
\end{equation*}
\end{proof}

We will also need the following version of Proposition \ref{1dXY} with
Dirichlet boundary conditions. The proof follows from essentially the same
arguments as for Proposition \ref{1dXY} --- see also Remark \ref{heat} for
the last statement of the following proposition.

\begin{proposition}
\label{1dXY'} Consider the XY model on\ $G_{n}$ with Dirichlet boundary
condition, defined by the Gibbs measure%
\begin{equation*}
Z_{G_{n}}^{-1}\left( \theta _{0},\theta _{1}\right) \exp \left( -\mathcal{B}%
_{n}H\left( S^{n}\right) \right) \prod_{i\in \mathcal{V}(G_{n})\backslash
\left\{ 0,1\right\} }d\theta _{i}^{n}\delta \left( \theta _{0}^{n}=\theta
_{0}\right) \delta \left( \theta _{1}^{n}=\theta _{1}\right) ,
\end{equation*}%
for some $\theta _{0},\theta _{1}\in (-\pi ,\pi ]$, where $H\left(
S^{n}\right) $ is defined in (\ref{H1dXY}). If $\mathcal{B}_{n}/n\rightarrow
b\in \left( 0,\infty \right) $, then $S_{\left[ nt\right] }^{n}$ converges
in distribution to $S\left( t\right) =\exp \left( iB\left( t\right) /\sqrt{b}%
\right) $, where $B\left( t\right) $ is the one dimensional Brownian bridge
conditioned on $\exp \left( iB\left( 0\right) /\sqrt{b}\right) =e^{i\theta
_{0}}$, and $\exp \left( iB\left( 1\right) /\sqrt{b}\right) =e^{i\theta
_{1}} $. Moreover, given any $\theta _{0},\theta _{1},\theta _{0}^{\prime
},\theta _{1}^{\prime }\in (-\pi ,\pi ]$, as $n\rightarrow \infty $,%
\begin{equation*}
\frac{Z_{G_{n}}\left( \theta _{0},\theta _{1}\right) }{Z_{G_{n}}\left(
\theta _{0}^{\prime },\theta _{1}^{\prime }\right) }\rightarrow \frac{%
\sum_{m\in \mathbb{Z}}\exp \left( -\frac{1}{2b}\left( \theta _{1}-\theta
_{0}+2\pi m\right) ^{2}\right) }{\sum_{m\in \mathbb{Z}}\exp \left( -\frac{1}{%
2b}\left( \theta _{1}^{\prime }-\theta _{0}^{\prime }+2\pi m\right)
^{2}\right) }.
\end{equation*}
\end{proposition}

\section{Proof of Theorem \protect\ref{V}\label{Sec:pf}}

We now apply Theorem \ref{weak}, Proposition \ref{1dXY} and Proposition \ref%
{1dXY'} to prove Theorem \ref{V}. We first prove Theorem \ref{V} when $G$ is
a single edge, i.e., $G=\left( \left\{ x,y\right\} ,\left\{ e\right\}
\right) $, and without loss of generality we can identify $x$ and $y$ with $%
0 $ and $1$ in the unit interval $\left[ 0,1\right] $. Consider an XY\ model
on $G_{n}\mathbb{\ }$of Section \ref{Sec:V}, defined by the Gibbs measure (%
\ref{xy}) with $\mathcal{B}=\mathcal{B}_{n}=nJ_{e}^{-1}$. By Proposition \ref%
{1dXY}, $\left( S_{0}^{n},S_{n}^{n}\right) $ converges in distribution to $%
\left( \exp \left( i\Theta \left( 0\right) \right) ,\exp \left( i\Theta
\left( 1\right) \right) \right) $, where $\Theta \left( 0\right) $ is
uniform in $[-\pi ,\pi )$ and $\exp \left( i\left( \Theta \left( 1\right)
-\Theta \left( 0\right) \right) \right) $ is distributed as $\exp \left(
iB\left( 1\right) \sqrt{J_{e}}\right) $. In other words, $\left( \exp \left(
i\Theta \left( 0\right) \right) ,\exp \left( i\Theta \left( 1\right) \right)
\right) $ has the probability density%
\begin{equation*}
Z^{-1}V_{e}\left( \theta _{1}-\theta _{0}\right) d\theta _{0}d\theta _{1},%
\text{ \ }\theta _{0},\theta _{1}\in (-\pi ,\pi ]
\end{equation*}%
which has the same distribution as for the Villain model on $G$. By Theorem %
\ref{LYXY}, for all $\lambda _{0},\lambda _{1}\geq 0$, the distribution of 
\begin{equation*}
\lambda _{0}\cos \theta _{0}^{n}+\lambda _{1}\cos \theta _{n}^{n}
\end{equation*}%
satisfies the Lee-Yang property (Property (3) of Definition \ref{LYp}). It
also satisfies Properties (1) and (2). Applying Theorem \ref{weak}, we
conclude that%
\begin{equation*}
\lambda _{0}\cos \Theta \left( 0\right) +\lambda _{1}\cos \Theta \left(
1\right)
\end{equation*}%
is also of Lee-Yang type. This proves Theorem \ref{V} in the special case
when $G$ is a single edge.

In the general case (see Figure $1$ for an illustration), given $G=\left( 
\mathcal{V}\left( G\right) ,\mathcal{E}\left( G\right) \right) $ we first
replace each vertex $v\in $ $\mathcal{V}\left( G\right) $ by several new
vertices --- one denoted $v^{\ast }$ and then one more denoted $\left(
v,e,0\right) $ for each $e\in \mathcal{E}\left( G\right) $ incident on $v,$
which we write as $e\sim v$. (We will also denote $\left( v,e,0\right) $ for 
$e=\left\{ v,w\right\} $ by $\left( w,e,n\right) $.) We then create one new
edge between $\left\{ v^{\ast },\left( v,e,0\right) \right\} _{v\in \mathcal{%
V}}$ and each $\left( v,e,0\right) $. Each $e=\left\{ v,w\right\} \in 
\mathcal{E}\left( G\right) $ is replaced by a collection $\mathcal{V}\left(
e,n\right) $ of $n-1$ new vertices which will be labelled $\left(
v,e,1\right) ,\left( v,e,2\right) ,...,\left( v,e,n-1\right) $ (and also in
opposite order $\left( w,e,1\right) ,\left( w,e,2\right) ,...,\left(
w,e,n-1\right) $) and a collection $\mathcal{E}\left( e,n\right) $ of $n$
new edges: $\left\{ \left( v,e,0\right) ,\left( v,e,1\right) \right\}
,\left\{ \left( v,e,1\right) ,\left( v,e,2\right) \right\} ,...$ (or in
opposite order $\left\{ \left( w,e,0\right) ,\left( w,e,1\right) \right\}
,...$). This defines a new graph $G_{n}^{\ast }=\left( \mathcal{V}\left(
G_{n}^{\ast }\right) ,\mathcal{E}\left( G_{n}^{\ast }\right) \right) $.

\begin{figure}[htbp]
\centering
\includegraphics[width=.8\linewidth]{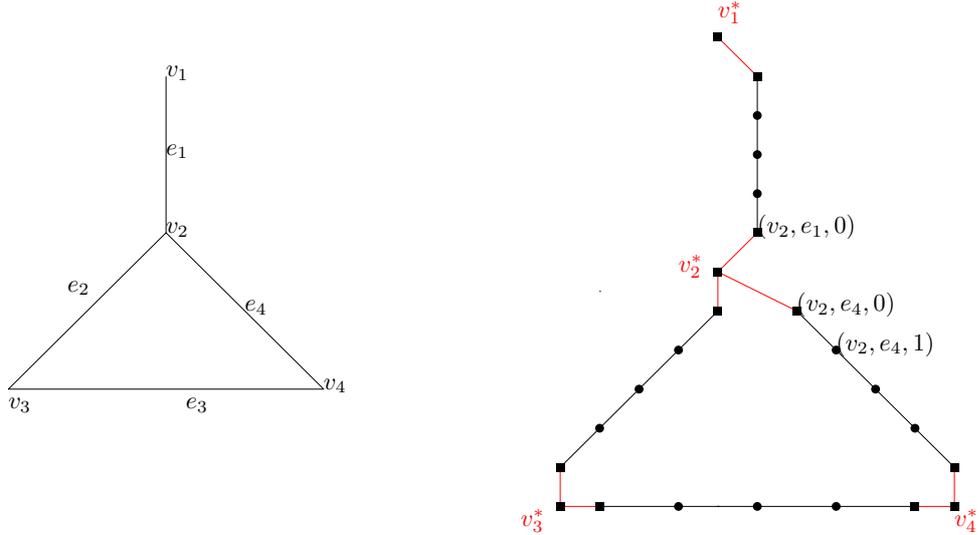}
\caption{An example of a graph $G$ (on the left) with $4$ vertices and $4$
edges and the modified graph $G_{n}^{\ast }$ with $n=4$. Square dots in $%
G_{n}^{\ast }$ indicate the replacements for the vertices of $G$; some of
these are labelled as discussed in the text. Circular dots indicates the $%
n-1=3$ new vertices replacing each edge of $G$.}
\label{fig:1}
\end{figure}

As indicated by Figure $1$, to prove Theorem \ref{V}, it suffices to prove
that for all non-negative $\left\{ \lambda _{v}\right\} _{v\in \mathcal{V}%
\left( G\right) }$%
\begin{eqnarray}
&&Z_{G}\mathbb{E}\left[ \exp \left( z\sum_{v\in \mathcal{V}\left( G\right)
}\lambda _{v}\cos \Theta _{v}\right) \right]   \notag \\
&=&\int \exp \left( z\sum_{v\in \mathcal{V}\left( G\right) }\lambda _{v}\cos
\theta _{v}\right) \prod_{e=\left( u,v\right) \in \mathcal{E}\left( G\right)
}V_{e}\left( \theta _{u}-\theta _{v}\right) \prod_{v\in \mathcal{V}\left(
G\right) }d\theta _{v}  \label{PIZV}
\end{eqnarray}%
has only pure imaginary zeroes.

For any fixed $J>0$ and $\left\{ J_{e}\right\} _{e\in \mathcal{E}\left(
G\right) }$, let $\mathcal{B}_{e}=nJ_{e}^{-1}$, and consider an
(inhomogeneous) XY model on $G_{n}^{\ast }$, defined by the Gibbs measure 
\begin{equation*}
Z_{G_{n}^{\ast }}^{-1}\exp \left( -H_{XY}^{J}\left( \theta ^{n}\right)
\right) \prod_{i\in \mathcal{V}\left( G_{n}^{\ast }\right) }d\theta _{i}^{n},%
\text{ \ }\theta _{i}^{n}\in (-\pi ,\pi ]\text{ }\forall i,
\end{equation*}%
with Hamiltonian 
\begin{eqnarray*}
H_{XY}^{J}\left( \theta ^{n}\right) &:&=\sum_{e\in \mathcal{E}\left(
G\right) }H_{XY}^{e,n}\left( \theta ^{n}\right) -J\sum_{v\in \mathcal{V}%
\left( G\right) }\sum_{e\sim v}\cos \left( \theta _{v^{\ast }}^{n}-\theta
_{\left( v,e,0\right) }^{n}\right) \\
&:&=-\sum_{e\in \mathcal{E}\left( G\right) }\sum_{\left( i,j\right) \in 
\mathcal{E}\left( e,n\right) }\mathcal{B}_{e}\cos \left( \theta
_{i}^{n}-\theta _{j}^{n}\right) -J\sum_{v\in \mathcal{V}\left( G\right)
}\sum_{e\sim v}\cos \left( \theta _{v^{\ast }}^{n}-\theta _{\left(
v,e,0\right) }^{n}\right) .
\end{eqnarray*}%
Applying the Lee-Yang theorem for the XY\ model (Theorem \ref{LYXY}), we see
that for all non-negative $\left\{ \lambda _{v^{\ast }}\right\} _{v\in 
\mathcal{V}\left( G\right) }$, 
\begin{equation*}
\int \exp \left( -H_{XY}^{J}\left( \theta ^{n}\right) +z\sum_{v\in \mathcal{V%
}\left( G\right) }\lambda _{v^{\ast }}\cos \theta _{v^{\ast }}^{n}\right)
\prod_{v\in \mathcal{V}\left( G_{n}^{\ast }\right) }d\theta _{v}^{n}
\end{equation*}%
has only pure imaginary zeros. Given $e=\left( u,v\right) \in \mathcal{E}%
\left( G\right) $, and any $\theta _{u},\theta _{v}\in (-\pi ,\pi ]$, we
also define 
\begin{equation*}
Z_{XY}^{e,n}\left( \theta _{u},\theta _{v}\right) =\int \exp \left(
-H_{XY}^{e,n}\left( \theta ^{n}\right) \right) \prod_{j\in \mathcal{V}\left(
e,n\right) }d\theta _{j}^{n}\delta \left( \theta _{u,e,0}^{n}=\theta
_{u}\right) \delta \left( \theta _{v,e,0}^{n}=\theta _{v}\right) .
\end{equation*}%
Using the definition of $H_{XY}^{J}$, we see that 
\begin{eqnarray}
&&\int \exp \left( -\sum_{e\in \mathcal{E}\left( G\right)
}H_{XY}^{e,n}\left( \theta ^{n}\right) +J\sum_{v\in \mathcal{V}\left(
G\right) }\sum_{e\sim v}\cos \left( \theta _{v^{\ast }}^{n}-\theta _{\left(
v,e,0\right) }^{n}\right) +z\sum_{v\in \mathcal{V}\left( G\right) }\lambda
_{v^{\ast }}\cos \theta _{v^{\ast }}^{n}\right)  \notag \\
&&\times \prod_{v\in \mathcal{V}\left( G_{n}^{\ast }\right) }d\theta _{v}^{n}%
\text{ }\times \left( \prod_{e\in \mathcal{E}\left( G\right)
}Z_{XY}^{e,n}\left( 0,0\right) \right) ^{-1}  \label{xunorm}
\end{eqnarray}%
has only pure imaginary zeros.

Applying Proposition \ref{1dXY'}, by taking $n\rightarrow \infty $ we have%
\begin{eqnarray*}
\frac{Z_{XY}^{e,n}\left( \theta _{u},\theta _{v}\right) }{Z_{XY}^{e,n}\left(
0,0\right) } &\rightarrow &\text{ Constant }\cdot \sum_{m\in \mathbb{Z}}\exp
\left( -\frac{J_{e}}{2}\left( \theta _{u}-\theta _{v}+2\pi m\right)
^{2}\right) \\
&=&\text{ Constant }\cdot V_{e}\left( \theta _{u}-\theta _{v}\right) .
\end{eqnarray*}%
Therefore, omitting all the superscripts $n$ in $\theta _{\cdot }^{n}$, the
integral (\ref{xunorm}) converges to%
\begin{eqnarray}
&&\text{Constant }\cdot \int \exp \left( J\sum_{v\in \mathcal{V}\left(
G\right) }\sum_{e\sim v}\cos \left( \theta _{v^{\ast }}-\theta _{\left(
v,e,0\right) }\right) +z\sum_{v\in \mathcal{V}\left( G\right) }\lambda
_{v^{\ast }}\cos \theta _{v^{\ast }}\right)  \label{vil} \\
&&\times \prod_{e=\left( u,v\right) \in \mathcal{E}\left( G\right)
}V_{e}\left( \theta _{u}-\theta _{v}\right) \prod_{v\in \mathcal{V}\left(
G\right) }d\theta _{v^{\ast }}.  \notag
\end{eqnarray}%
By Theorem \ref{weak}, we see that (\ref{vil}), as a function of $z$, only
has pure imaginary zeroes. Applying Laplace's method, we can multiply (\ref%
{vil}) by the right $J-$dependent factor and let $J\rightarrow \infty $ to
recover the Villain model on $G$, finishing the proof of (\ref{PIZV}).

\section{Complex Gaussian Multiplicative Chaos\label{Sec:GMC}}

\subsection{Continuum Complex Gaussian multiplicative chaos\label{GMC}}

In this subsection we apply Theorem \ref{slowtail2} to complex Gaussian
multiplicative chaos. Let $D_{r}\subset \mathbb{C}$, $r>0$ be the disk of
radius $r$ centered at the origin. For $\beta >0$, complex Gaussian
multiplicative chaos in $D_{r}$ is defined as $\exp \left( i\beta
h^{r}\right) $, where $h^{r}$ is a Dirichlet zero boundary condition GFF in $%
D_{r}$. As was discussed in \cite{LRV}, for $\beta \in \left( 0,\sqrt{2}%
\right) $ the (complex-valued) measure $e^{i\beta h^{r}\left( x\right) }dx$
is well-defined and is absolutely continuous with respect to Lebesgue
measure.

We will obtain results for complex Gaussian multiplicative chaos in the
whole plane, which is defined informally as $\exp \left( i\beta h\right) $,
where $h$ is a GFF\ in all of $\mathbb{R}^{2}$. Mathematically, the GFF\ in $%
\mathbb{R}^{2}$ is not well-defined as a random generalized function --- it
is only defined up to an additive constant. Therefore, a priori, $\exp
\left( i\beta h\right) $ is only defined up to a multiplicative constant (on
the unit circle in $\mathbb{C}$). In this section, we obtain the measure $%
\exp \left( i\beta h\right) dx$ by defining for any bounded simply connected
domain $U\subset \mathbb{R}^{2}$, the random variable 
\begin{equation*}
\int_{U}\exp \left( i\beta h\left( x\right) \right) dx
\end{equation*}%
as the limit in distribution of 
\begin{equation*}
\int_{U}\exp \left( i\beta h^{r}\left( x\right) \right) dx
\end{equation*}%
as $r\rightarrow \infty $. The existence of the limit of the modulus of $%
\int_{U}\exp \left( i\beta h^{r}\left( x\right) \right) dx$ as $r\rightarrow
\infty $ is proved in the appendix of \cite{LSZW} (and an alternative proof
may be obtained by using a construction similar to that in \cite{LRV}). To
give a precise definition of the limiting field without analyzing
multiplicative factors, we may simply proceed as follows. Let $\Phi $ be a
random variable uniformly distributed on $(-\pi ,\pi ]$ that is independent
of the entire collection of random fields $\left\{ h_{r}:r>0\right\} $. Then
define $\int_{U}\exp \left( i\beta h\left( x\right) \right) dx$ as the $%
r\rightarrow \infty $ limit in distribution of 
\begin{equation*}
e^{i\Phi }\int_{U}\exp \left( i\beta h^{r}\left( x\right) \right) dx.
\end{equation*}

Based on the spin-wave picture of the XY and Villain models, it a priori
seems reasonable to conjecture that complex Gaussian multiplicative chaos
with any of the boundary conditions mentioned in\ Remark \ref{bc} also
satisfies a Lee-Yang property. Here we restrict attention first to the
whole-plane field (without boundary). Let $U\subset \mathbb{R}^{2}$ be a
bounded simply connected domain with smooth boundary, and for any
non-negative bounded continuous function $\lambda :D_{r}\rightarrow \mathbb{R%
}$, define the random variable%
\begin{equation}
X_{\infty }=\int_{U}\lambda \left( x\right) \Re{e^{i\beta h\left( x\right) }}%
dx.  \label{Xinf}
\end{equation}%
If the Lee-Yang property were valid, then $\mathbb{E}\left[ \exp \left(
zX_{\infty }\right) \right] $ as a function of $z$\ would only have pure
imaginary zeroes. The next theorem \textit{disproves }the Lee-Yang property
for complex Gaussian multiplicative chaos, when $\beta \in \left( 1,\sqrt{2}%
\right) $.

\begin{proposition}
\label{wpGMC} Let $X_{\infty }$ be as in (\ref{Xinf}) with $\lambda \left(
x\right) \equiv 1$ for any bounded simply connected domain $U$ with smooth
boundary. Then for any $\beta \in \left( 1,\sqrt{2}\right) $, $\mathbb{E}%
\left[ \exp \left( zX_{\infty }\right) \right] $ has some zeroes that are
not purely imaginary.
\end{proposition}

\bigskip To prove Proposition \ref{wpGMC}, we start with the following tail
estimate for 
\begin{equation*}
W_{U}:=\int_{U}e^{i\beta h\left( x\right) }dx.
\end{equation*}

\begin{proposition}
\label{tailwp}For $\beta \in (0,\sqrt{2})$ and $k\in \mathbb{N}$, 
\begin{equation}
\log \mathbb{E}\left\vert W_{U}\right\vert ^{2k}=\beta ^{2}k\log k+c\left(
\beta ,U\right) k+o\left( k\right) \text{ \ as }k\rightarrow \infty \text{.}
\label{mom}
\end{equation}%
Therefore, for $t$ sufficiently large 
\begin{equation}
\mathbb{P}\left( \left\vert W_{U}\right\vert >t\right) =\exp \left( -c^{\ast
}\left( \beta ,U\right) t^{\frac{2}{\beta ^{2}}}+o(t^{\frac{2}{\beta ^{2}}%
})\right) .  \label{tailwp1}
\end{equation}%
In (\ref{mom}) and (\ref{tailwp1}) $c\left( \beta ,U\right) $ and $c^{\ast
}\left( \beta ,U\right) $ in $\left( 0,\infty \right) $ depend on $\beta $
and $U$ but not on $k$.
\end{proposition}

\begin{proof}[Proof of Proposition \protect\ref{tailwp}]
Apply Lemma A.1 of \cite{LSZW}, Appendix A (the same calculation works for
any domain $U\subset \mathbb{R}^{2}$), to obtain%
\begin{equation*}
\mathbb{E}\left\vert W_{U}\right\vert ^{2k}=\int_{U^{\otimes 2k}}\left( 
\frac{\prod_{1<i<j<k}\left\vert x_{i}-x_{j}\right\vert \left\vert
y_{i}-y_{j}\right\vert }{\prod_{i,j}\left\vert x_{i}-y_{j}\right\vert }%
\right) ^{\beta ^{2}}d\vec{x}d\vec{y}.
\end{equation*}%
This is the partition function for a Coulomb gas ensemble with $k$ positive
charges and $k$ negative charges confined in $U$. Eq (\ref{mom}) then
follows from the calculation of this partition function in \cite{GP} (see
also \cite{LSZW}). Eq. (\ref{tailwp1}) is a consequence of (\ref{mom}) and
Chebyshev's inequality (see Lemma A.2 of \cite{LSZW} for a proof when $%
U=D_{1}$).
\end{proof}

\begin{proof}[Proof of Proposition \protect\ref{wpGMC}]
We now apply Proposition \ref{tailwp} and Theorem \ref{slowtail2} to finish
the proof of Proposition \ref{wpGMC}. To apply Proposition \ref{tailwp}, we
want upper and lower bounds on the moment generating function and
distribution tails for 
\begin{equation*}
X_{\infty }=\Re{W_{U}}=\int_{U}\Re{e^{i\beta h\left( x\right) }}dx.
\end{equation*}%
and for powers of $\left\vert X_{\infty }\right\vert $ in terms of those for 
$\left\vert W_{U}\right\vert $.

The upper bounds are immediate since 
\begin{equation*}
X_{\infty }\leq \left\vert X_{\infty }\right\vert \leq \left\vert
W_{U}\right\vert ,
\end{equation*}%
so by Proposition \ref{tailwp}, 
\begin{equation*}
\mathbb{P}\left( \left\vert X_{\infty }\right\vert >t\right) \leq \mathbb{P}%
\left( \left\vert W_{U}\right\vert >t\right) \leq \exp \left( -0.9c^{\ast
}\left( \beta ,U\right) t^{2/\beta ^{2}}\right)
\end{equation*}%
for $t$ large, which implies, by an explicit computation that for $b>0$, 
\begin{equation*}
\mathbb{E}e^{b\left\vert X_{\infty }\right\vert ^{2/\beta ^{2}}}<\infty .
\end{equation*}%
Note that for $\beta \in \left( 1,\sqrt{2}\right) $, $2/\beta ^{2}\in \left(
1,2\right) $.

For a lower bound, we use the fact that $W_{U}$ is equidistributed with $%
e^{i\phi }W_{U}$ for any real $\phi $, which implies that $X_{\infty }= \Re{%
W_{U}}$ and $Y_{\infty }=\Im{W_{U}}$ are equidistributed. Thus 
\begin{equation}
\mathbb{E}e^{b^{\prime }X_{\infty }^{2}}=\mathbb{E}\left( e^{b^{\prime
}X_{\infty }^{2}}+e^{b^{\prime }Y_{\infty }^{2}}\right) /2\geq \mathbb{E}%
\left( e^{b^{\prime }X_{\infty }^{2}/2}e^{b^{\prime }Y_{\infty
}^{2}/2}\right) =\mathbb{E}\left( e^{b^{\prime }\left\vert W_{U}\right\vert
^{2}/2}\right) .  \label{slow}
\end{equation}%
Then by Proposition \ref{tailwp}, $\mathbb{E}\left( e^{b^{\prime }\left\vert
W_{U}\right\vert ^{2}/2}\right) =\infty $ for any $b^{\prime }>0$ and now by
Theorem \ref{slowtail2} we conclude that $\mathbb{E}\left[ \exp \left(
zX_{\infty }\right) \right] $ must have some zeroes that are not purely
imaginary.
\end{proof}

\subsection{ Discrete complex Gaussian multiplicative chaos}

Let $G=\left( \mathcal{V},\mathcal{E}\right) $ be a planar domain, i.e., $%
\mathcal{V}\subset \mathbb{Z}^{2}$ and $\mathcal{E}$ is the set of nearest
neighbor edges between vertices of $\mathcal{V}$. We use $\partial G$ to
denote the boundary vertices of $G$, i.e., those that are connected to $%
\mathbb{Z}^{2}\setminus \mathcal{V}$ by one or more edges. Consider a random
function $h\left( i\right) $, $i\in \mathcal{V}$, such that $h\left(
i\right) \in (-\pi ,\pi ]$, distributed according to the (conditional on $%
\phi $) joint density%
\begin{eqnarray}
&&Z^{-1}\sum_{i\in \mathcal{V\setminus }\partial G}\sum_{n_{i}\in \mathbb{Z}%
}\prod_{\left( i,j\right) \in \mathcal{E}}\exp \left( -\frac{B_{ij}}{2}%
(h\left( i\right) -2\pi n_{i}-h\left( j\right) +2\pi n_{j})^{2}\right)
\label{gmc} \\
&&\times \prod_{i\in \mathcal{V\setminus }\partial G}dh\left( i\right)
\prod_{i\in \partial G}\delta \left( h\left( i\right) -\phi \right) ,  \notag
\end{eqnarray}%
where $\phi $ is the value of a uniformly distributed in $(-\pi ,\pi ]$
random variable $\Phi $ and $\left\{ B_{ij}\right\} $ are some positive
numbers. Heuristically, one can construct $h$ by first choosing $\phi $
uniformly in $(-\pi ,\pi ]$, then sampling a discrete Gaussian free field
(DGFF) on $G$ with Dirichlet boundary condition $\phi $ (denoted as $h^{\phi
}$); finally, we obtain $h$ by taking the value of $h^{\phi }$ modulo $2\pi $%
. We call the random function $h$ thus defined the discrete complex Gaussian
multiplicative chaos on $G$, since it is a discrete analogue of the complex
Gaussian multiplicative chaos considered in Section \ref{GMC}.

Motivated by the spin-wave picture of the XY and Villain models, it is
natural to ask whether the discrete complex Gaussian multiplicative chaos
with any boundary conditions mentioned in\ Remark \ref{bc} also satisfies a
Lee-Yang theorem. Here we restrict attention to Dirichlet boundary
conditions, with $h$ sampled from the discrete complex Gaussian
multiplicative chaos (\ref{gmc}), and we define 
\begin{equation*}
M=\sum_{i\in \mathcal{V}}\lambda _{i}\cos \left( h\left( i\right) \right) ,%
\text{ with }\lambda _{i}\geq 0\text{.}
\end{equation*}%
This is a discrete analogue of the random variable $X_{\infty }$ defined in (%
\ref{Xinf}). In this section we show that the Lee-Yang property cannot be
valid in general, by giving a family of counter-examples in the next
proposition.

We first introduce some notation. Recall that $D_{r}$ is the disk of radius $%
r$ centered at the origin. Given $x\in D_{r}$, we denote by $C_{r}\left(
x\right) $ the conformal radius of $D_{r}$ from $x$.

\begin{proposition}
\label{noLY}Take $B_{ij}=\beta ^{-2},$ $\mathcal{G}_{n,r}=D_{nr}\cap \mathbb{%
Z}^{2}$ and $g=2/\pi $. Let $h^{n,r}$ be sampled from (\ref{gmc}) on $%
\mathcal{G}_{n,r}$ and 
\begin{equation*}
M_{n,r}=\sum_{x\in D_{n}\cap \mathbb{Z}^{2}}\frac{1}{n^{2}}\Re \left[ \exp
\left( ih^{n,r}\left( x\right) +\frac{\beta ^{2}}{2}\left[ g\log
n-gC_{r}\left( \frac{x}{n}\right) \right] \right) \right] .
\end{equation*}%
Then for $\beta \in \left( 1,\sqrt{2}\right) $, large enough $r$ and then
large enough $n$, $\mathbb{E}\left[ \exp \left( zM_{n,r}\right) \right] $
has some zeros that are not purely imaginary.
\end{proposition}

Notice that we can write $M_{n,r}=\sum_{x\in D_{n}\cap \mathbb{Z}%
^{2}}\lambda _{n,r}\left( x\right) \cos \left( h^{n,r}\left( x\right)
\right) $, with 
\begin{equation*}
\lambda _{n,r}\left( x\right) =n^{\frac{\beta ^{2}}{2}g-2}e^{-gC_{r}\left(
x/n\right) }.
\end{equation*}

We start by proving the following limit theorem for $M_{n,r}$.

\begin{lemma}
\label{limit}For any $r>0$, as $n\rightarrow \infty $, $M_{n,r}$ converges
in distribution to 
\begin{equation}
\int_{D_{1}}\Re \left[ e^{i\left( \beta h^{r}\left( x\right) +\Phi \right) }%
\right] dx,  \label{regmc}
\end{equation}%
where $h^{r}$ is a continuum GFF on $D_{r}$ with zero boundary, and $\Phi $
is uniformly distributed in $(-\pi ,\pi ]$.
\end{lemma}

\begin{proof}
Recall that for any $c>0$, both the discrete and continuum GFF have the
Gibbs-Markov property: a GFF/DGFF with Dirichlet boundary condition $c$
equals in law a zero boundary GFF/DGFF plus the constant $c$. Therefore it
suffices to prove the lemma with $\Phi =0$ in both (\ref{gmc}) and (\ref%
{regmc}). Next define 
\begin{equation*}
\hat{M}_{n,r}=\sum_{x\in D_{n}\cap \mathbb{Z}^{2}}\frac{1}{n^{2}}\exp \left(
ih^{n,r}\left( x\right) +\frac{\beta ^{2}}{2}\left[ g\log n-gC_{r}\left( 
\frac{x}{n}\right) \right] \right) .
\end{equation*}%
Since $M_{n,r}=\Re \left[ \hat{M}_{n,r}\right] $ is a bounded continuous
function of $\hat{M}_{n,r}$, it suffices to show that $\hat{M}_{n,r}$
converges in distribution to 
\begin{equation*}
\int_{D_{1}}e^{i\beta h^{r}\left( x\right) }dx.
\end{equation*}

To see this, we compute the moments of $\hat{M}_{n,r}$ as%
\begin{equation}
\mathbb{E}\left[ \hat{M}_{n,r}^{k}\right] =\sum_{x_{1},...,x_{k}\in
D_{n}\cap \mathbb{Z}^{2}}\frac{1}{n^{2k}}\mathbb{E}\left[
\prod_{i=1}^{k}e^{ih^{n,r}\left( x_{i}\right) }\right] \exp \left( \frac{%
\beta ^{2}}{2}\left[ gk\log n-g\sum_{i=1}^{k}C_{r}\left( \frac{x_{i}}{n}%
\right) \right] \right) .  \label{Mmom}
\end{equation}%
Let $G_{nr}$ be the Dirichlet Green's function on $D_{nr}\cap \mathbb{Z}^{2}$%
. We have 
\begin{eqnarray*}
&&\mathbb{E}\left[ \prod_{i=1}^{k}e^{ih^{n,r}\left( x_{i}\right) }\right] =%
\mathbb{E}\left[ \exp \left( i\sum_{i=1}^{k}h^{n,r}\left( x_{i}\right)
\right) \right] \\
&=&\exp \left( -\frac{1}{2}\text{Var}\left[ \sum_{i=1}^{k}h^{n,r}\left(
x_{i}\right) \right] \right) \\
&=&\exp \left( -\frac{\beta ^{2}}{2}\left[ \sum_{i=1}^{k}G_{nr}\left(
x_{i},x_{i}\right) +\sum_{i\neq j}G_{nr}\left( x_{i},x_{j}\right) \right]
\right) .
\end{eqnarray*}%
We now use the standard estimates for the lattice Green's function (see,
e.g., \cite{Law}, Chapter 1.6), 
\begin{eqnarray*}
G_{nr}\left( x_{i},x_{i}\right) &=&g\log n-gC_{r}\left( \frac{x_{i}}{n}%
\right) +O\left( 1/n\right) , \\
G_{nr}\left( x_{i},x_{j}\right) &=&g^{D_{r}}\left( x_{i},x_{j}\right)
+O\left( 1/n\right) ,
\end{eqnarray*}%
where $g^{D_{r}}$ is the Dirichlet Green's function in $D_{r}$ with zero
boundary condition. Substituting these into (\ref{Mmom}) and using the
convergence of Riemann sums to integrals concludes the proof.
\end{proof}

\begin{proof}[Proof of Proposition \protect\ref{noLY}]
Notice that when $B_{ij}=\beta ^{-2}$, Lemma \ref{limit} implies that as $%
n\rightarrow \infty $, $M_{n,r}$ converges in distribution to the integral 
\begin{equation*}
X_{r}:=\int_{D_{1}}\Re \left[ e^{i\left( \beta h^{r}\left( x\right) +\Phi
\right) }\right] dx.
\end{equation*}%
Moreover, by the construction and the rotational invariance of the whole
plane complex Gaussian multiplicative chaos, as $r\rightarrow \infty $, $%
\int_{D_{1}}e^{i\left( \beta h^{r}\left( x\right) +\Phi \right) }dx$
converges in distribution to $W_{D_{1}}e^{i\Phi }$, which is equidistributed
with $W_{D_{1}}$. Thus $X_{r}$ converges in distribution to $X_{\infty }$.
However, by Eq (\ref{slow}), when $\beta >1$, $\mathbb{E}\left(
e^{b\left\vert X_{\infty }\right\vert ^{2}}\right) =\infty $ for all $b>0$.
Applying Corollary \ref{seq} completes the proof.
\end{proof}

\bigskip

\paragraph{\textbf{Acknowledgments:}}

The research reported here was supported in part by U.S. NSF grant
DMS-1507019. We thank P.-F. Rodriguez, T. Spencer and a very conscientious
anonymous referee for useful communications and suggestions.

\bibliographystyle{plain}
\bibliography{acompat,LY}

\newif\ifabfull\abfulltrue
\begin{thebibliography}{10}

\bibitem{DH}
Jonathan Dimock and Thomas Hurd.
\newblock Sine-Gordon revisited.
\newblock {\em Annales Henri Poincar{\'e}}, 1(3): 499-541, 2000.

\bibitem{DN}
Fran{\c{c}}ois Dunlop and Charles M Newman. 
\newblock Multicomponent field theories and classical rotators.
\newblock {\em Communications in Mathematical Physics}, 44(3): 223-235, 1975.

\bibitem{Dun2}
Fran{\c{c}}ois Dunlop.
\newblock Zeros of the partition function and {G}aussian inequalities for the
  plane rotator model.
\newblock {\em Journal of Statistical Physics}, 21(5):561--572, 1979.

\bibitem{Dys}
Freeman~J Dyson.
\newblock General theory of spin-wave interactions.
\newblock {\em Physical Review}, 102(5):1217, 1956.

\bibitem{Erd}
Arthur Erdelyi
\newblock Asymptotic expansions. 
\newblock {\em Dover Publicaitons}, 1956.


\bibitem{FR}
J{\"u}rg Fr{\"o}hlich and Pierre-Fran{\c{c}}ois Rodriguez.
\newblock Some applications of the {L}ee-{Y}ang theorem.
\newblock {\em Journal of Mathematical Physics}, 53(9):095218, 2012.

\bibitem{FR2}
J{\"u}rg Fr{\"o}hlich and Pierre-Fran{\c{c}}ois Rodriguez.
\newblock On cluster properties of classical ferromagnets in an external magnetic field
\newblock {\em Journal of Statistical Physics}, 166(3-4), 828-840, 2017.

\bibitem{FS0}
J{\"u}rg Fr{\"o}hlich and Thomas Spencer.
\newblock On the statistical mechanics of classical {C}oulomb and dipole gases.
\newblock {\em Journal of Statistical Physics}, 24(4), 617-701, 1981.


\bibitem{FS}
J{\"u}rg Fr{\"o}hlich and Thomas Spencer.
\newblock The {K}osterlitz-{T}houless transition in two-dimensional abelian
  spin systems and the {C}oulomb gas.
\newblock {\em Communications in Mathematical Physics}, 81(4):527--602, 1981.

\bibitem{GO}
Anatoly Goldberg and Iossif Ostrovskii.
\newblock On the growth of entire ridge functions.
\newblock {\em Math. Physics and Functional Analysis}, Akad. Nauk Ukr. SSR, Fiz. Tehn. Inst. Nizkih Temperatur, Kharkov 5:3-10, 1974.

\bibitem{Gr}
Robert~B Griffiths.
\newblock Rigorous results for {I}sing ferromagnets of arbitrary spin.
\newblock {\em Journal of Mathematical Physics}, 10(9):1559--1565, 1969.


\bibitem{GRS}
Francesco Guerra, Lon Rosen, and Barry Simon.
\newblock Correlation inequalities and the mass gap in $p(\varphi)_2$.
\newblock {\em Communications in Mathematical Physics}, 41(1):19--32, 1975.

\bibitem{GP}
J~Gunson and LS~Panta.
\newblock Two-dimensional neutral {C}oulomb gas.
\newblock {\em Communications in Mathematical Physics}, 52(3):295--304, 1977.

\bibitem{Kal}
Olav Kallenberg.
\newblock Foundations of modern probability.
\newblock {\em Springer Science and Business Media}, 2006.

\bibitem{LRV}
Hubert Lacoin, R{\'e}mi Rhodes, and Vincent Vargas.
\newblock Complex gaussian multiplicative chaos.
\newblock {\em Communications in Mathematical Physics}, 337(2):569--632, 2015.

\bibitem{LSZW}
Thomas Lebl{\'e}, Sylvia Serfaty, and Ofer Zeitouni, with an appendix
  by Wei Wu.
\newblock Large deviations for the two-dimensional two-component plasma.
\newblock {\em Communications in Mathematical Physics},350(1): 301-360, 2017

\bibitem{LY}
Tsung-Dao Lee and Chen-Ning Yang.
\newblock Statistical theory of equations of state and phase transitions. {II}.
  {L}attice gas and {I}sing model.
\newblock {\em Physical Review}, 87(3):410, 1952.

\bibitem{LS}
Elliott~H Lieb and Alan~D Sokal.
\newblock A general {L}ee-{Y}ang theorem for one-component and multicomponent
  ferromagnets.
\newblock {\em Communications in Mathematical Physics}, 80(2):153--179, 1981.

\bibitem{L}
Eugene Lukacs.
\newblock Developments in Characteristic Function Theory. 
\newblock {\em Macmillian, New York}, 1983.

\bibitem{Law}
Gregory Lawler.
\newblock Intersection of random walks. 
\newblock {\em Springer}, 2013.


\bibitem{MW}
N~David Mermin and Herbert Wagner.
\newblock Absence of ferromagnetism or antiferromagnetism in one-or two-dimensional isotropic Heisenberg models.
\newblock {\em Physical Review Letters }, 17(22):1133, 1966.

\bibitem{Ne}
Charles~M Newman.
\newblock Zeros of the partition function for generalized {I}sing systems.
\newblock {\em Communications on Pure and Applied Mathematics}, 27(2):143--159,
  1974.

\bibitem{Ne2}
Charles~M Newman.
\newblock Inequalities for {I}sing models and field theories which obey the
  {L}ee-{Y}ang theorem.
\newblock {\em Communications in Mathematical Physics}, 41(1):1--9, 1975.

\bibitem{LP}
Oliver Penrose and Joel~L Lebowitz.
\newblock On the exponential decay of correlation functions.
\newblock {\em Communications in Mathematical Physics}, 39(3):165--184, 1974.

\bibitem{SG}
Barry Simon and Robert~B Griffiths.
\newblock The ($\varphi^4_ 2$) field theory as a classical {I}sing model.
\newblock {\em Communications in Mathematical Physics}, 33(2):145--164, 1973.

\bibitem{SF}
Masuo Suzuki and Michael~E Fisher.
\newblock Zeros of the partition function for the {H}eisenberg, ferroelectric,
  and general {I}sing models.
\newblock {\em Journal of Mathematical Physics}, 12(2):235--246, 1971.

\bibitem{Vi}
Jacques Villain.
\newblock Theory of one-and two-dimensional magnets with an easy magnetization plane. {II}. {T}he planar, classical, two-dimensional magnet.
\newblock {\em Journal de Physique}, 36(6):581--590, 1975.



\end{thebibliography}

\end{document}